\documentclass[a4paper,USenglish,cleveref]{lipics-v2019}

\usepackage{microtype} 
\usepackage{cite}
\usepackage{siunitx}
\usepackage{csquotes}
\usepackage{subcaption} 
\usepackage[basic]{complexity}
\usepackage[inline]{enumitem}
\usepackage{mathtools}

\graphicspath{{figures/}}

\usepackage{xcolor}

\title{Weak Unit Disk Contact Representations for Graphs without Embedding}

\author{Jonas Cleve}{Institut f\"ur Informatik, Freie Universit\"at Berlin}{jonascleve@inf.fu-berlin.de}{https://orcid.org/0000-0001-8480-1726}{Supported by ERC StG 757609.}
\authorrunning{J. Cleve}
\hideLIPIcs{}
\nolinenumbers{}

\newtheorem{obs}[theorem]{Observation}

\begin{document}

\maketitle

\begin{abstract}
    Weak unit disk contact graphs are graphs that admit representing nodes as a collection of internally disjoint unit disks whose boundaries touch if there is an edge between the corresponding nodes.
    In this work we focus on graphs without embedding, i.e., the neighbor order can be chosen arbitrarily.
    We give a linear time algorithm to recognize whether a caterpillar, a graph where every node is adjacent to or on a central path, allows a weak unit disk contact representation.
    On the other hand, we show that it is \NP-hard to decide whether a tree allows such a representation.
\end{abstract}

\section{Introduction}

A \emph{disk contact graph} \(G=(V,E)\) is a graph that has a geometric realization as a collection of internally disjoint disks mapped bijectively to the node set \(V\) such that two disks touch if and only if the corresponding nodes are connected by an edge in \(E\).
In an attempt to tackle the open problem of recognizing embedded caterpillars for disk contact graphs, \emph{weak} disk contact graphs were introduced, which allow two disks to touch even if they don't share an edge.
it was shown in this setting that the problem of recognizing embedded caterpillars is \NP-hard by Chiu, Cleve, and Nöllenburg~\cite{chiu-eurocg19}.
We continue this line of research by looking at graphs \emph{without} embedding.

\section{Recognizing Caterpillars in Linear Time}

Similar to the algorithm by Klemz, Nöllenburg and Prutkin~\cite{DBLP:conf/gd/KlemzNP15} we efficiently decide whether a caterpillar \(G=(V,E)\), a graph where every node is adjacent to or lies on a central path, admits a weak unit disk contact representation (WUDCR).
Let \(\Delta\) be the maximum degree of \(G\).
If \(\Delta\geq 7\) it is impossible to find a WUDCR\@: no unit disk can have more than six other adjacent unit disks.
For \(\Delta\leq 4\), \(G\) can even be realized as a (\emph{strong}) UDCR~\cite{DBLP:conf/gd/KlemzNP15}, which is also a WUDCR\@.
For \(5\leq\Delta\leq 6\) some caterpillars can be realized and some cannot; see \cref{fig:caterpillar-delta-5} for an example.
This can be formalized as the following

\begin{figure}[t]
    \centering
    \includegraphics[scale=1,page=1]{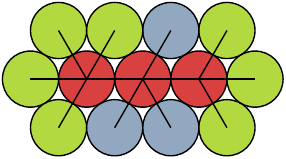}
    \caption{%
        A caterpillar with \(\Delta=5\), red (internal) nodes having degrees \(5\), \(5\), and \(4\).
        It becomes unrealizable when adding another child to the rightmost internal (red) node, giving it degree \(5\) as well.
    }%
    \label{fig:caterpillar-delta-5}
\end{figure}

\begin{figure}[t]
    \includegraphics[scale=1,page=6]{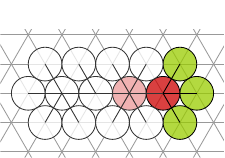}
    \hfill%
    \includegraphics[scale=1,page=7]{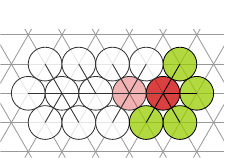}
    \hfill%
    \includegraphics[scale=1,page=8]{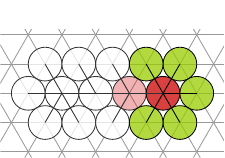}
    \caption{%
        The three different cases from \cref{lem:caterpillar-number-of-leaves}: \(d_k\leq 4\) (left), \(d_k=5\) (center), and \(d_k=6\) (right).
    }%
    \label{fig:caterpillar-algorithm-example}
\end{figure}

\begin{lemma}%
\label[lemma]{lem:caterpillar-number-of-leaves}
Let \(G\) be a caterpillar, \(v_0,v_1,\ldots,v_k,v_{k+1}\) a longest path in \(G\) and \(d_i=\deg(v_i)\) for all \(1\leq i\leq k\).
Then \(G\) has a WUDCR iff for all \(1\leq \ell\leq k\): \(\sum_{i=1}^{\ell}d_i\leq 4\ell+2\).
\end{lemma}
\begin{proof}[Proof (by induction)]
For \(k=1\) we have one node with \(d_1\) leaves.
The corresponding disk can have up to \(6\) neighboring disks and \(d_1\leq 4+2=6\).

Assuming the hypothesis holds for all \(\ell<k\) we show that it holds for \(k\).
Look at \cref{fig:caterpillar-algorithm-example} for a depiction of the cases.
\begin{enumerate}
    \item
        \(d_{k}\leq 4\): up to \(3\) new leaves are added, no overlap with previous disks needed.
        It follows that \(\sum_{i=1}^k d_i = \sum_{i=1}^{k-1}d_i + d_k \leq \sum_{i=1}^{k-1}d_i + 4 \stackrel{\text{IH}}{\leq} [4(k-1)+2]+4 = 4k+2\).
    \item
        \(d_k=5\): exactly \(4\) new leaves are added, taking away one position from the smaller construction, hence for \(k-1\) it must hold that \(\sum_{i=1}^{k-1}d_i \leq 4(k-1)+1\).
        It follows that \(\sum_{i=1}^k d_i = \sum_{i=1}^{k-1}d_i + d_k = \sum_{i=1}^{k-1}d_i + 5 \leq [4(k-1)+1]+5 = 4k+2\).
    \item
        \(d_k=6\): exactly \(5\) new leaves are added, taking away two position from the smaller construction, hence for \(k-1\) it must hold that \(\sum_{i=1}^{k-1}d_i \leq 4(k-1)\).
        It follows that \(\sum_{i=1}^k d_i = \sum_{i=1}^{k-1}d_i + d_k = \sum_{i=1}^{k-1}d_i + 6 \leq [4(k-1)]+6 = 4k+2\).
    \qedhere
\end{enumerate}
\end{proof}

\begin{theorem}\label[theorem]{thm:caterpillar-linear-time}
It can be decided in linear time whether a caterpillar \(G\) admits a WUDCR.
\end{theorem}
\begin{proof}
First determine a longest path \(v_0,v_1,\ldots,v_k,v_{k+1}\) in linear time.
Then check for all \(1\leq\ell\leq k\) whether \(\sum_{i=1}^{\ell}d_i\leq 4\ell+2\).
This linear number of sums is easily checked in linear time.
With \cref{lem:caterpillar-number-of-leaves} this immediately tells us whether a WUDCR for \(G\) exists.
\end{proof}

\section{NP-hardness of Recognizing Trees}%
\label{sec:np-hardness-for-trees}

Recognizing whether a tree has a WUDCR is \NP-hard---we use a reduction from Not-All-Equal-3SAT (NAE3SAT)~\cite{DBLP:conf/stoc/Schaefer78} via a \emph{logic engine construction}~\cite{DBLP:journals/ipl/BhattC87}.
An instance for the NAE3SAT problem is a 3SAT formula and a yes-instance is a formula \(\phi\) for which an assignment exists, which satisfies \(\phi\) and additionally contains at least one false literal per clause.
The logic engine construction, see \cref{fig:logic-engine-concept}, works as follows:
Given a formula with variables \(x_1,\dots,x_n\) and clauses \(c_1,\dots,c_m\) we construct an orthogonal drawing representing this formula.
We have one horizontal spine to which one pole (consisting of a thick positive and thin negative part) for each variable is attached at its center.
Each pole has \(m\) levels on the top and \(m\) on the bottom, each side representing the \(m\) clauses.
\begin{figure}[t]
  \begin{subfigure}[t]{0.46\linewidth}
    \centering
    \includegraphics[width=1\textwidth,page=1]{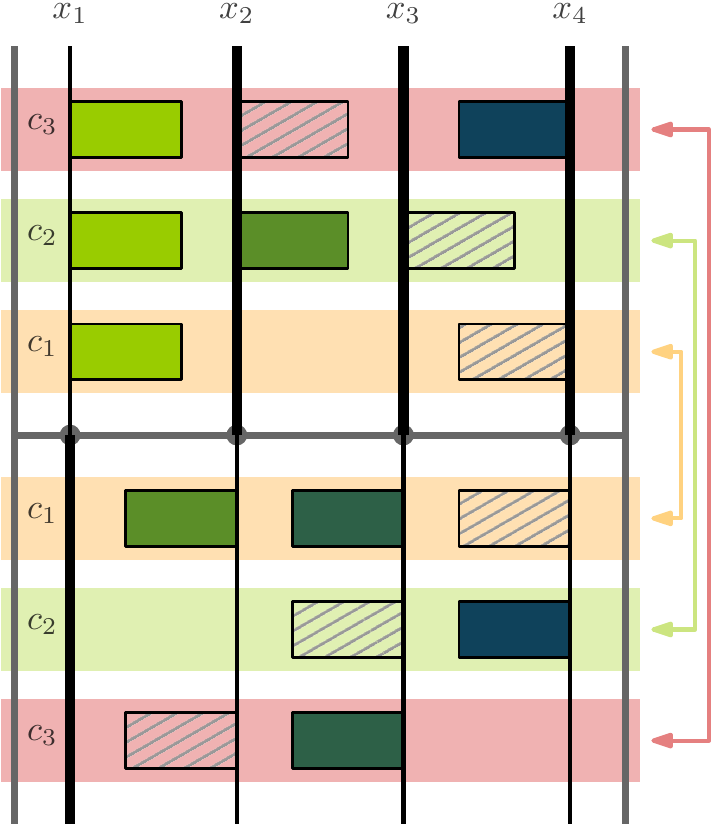}
    \caption{\(x_1=0\), \(x_2=1\), \(x_3=1\), and \(x_4=1\).}
  \end{subfigure}
  \hfill%
  \begin{subfigure}[t]{0.46\linewidth}
    \centering
    \includegraphics[width=1\textwidth,page=2]{logic-engine.pdf}
    \caption{\(x_1=1\), \(x_2=1\), \(x_3=0\), and \(x_4=1\).}
  \end{subfigure}
  \caption{%
      Two different logic engine realizations for the NAE3SAT formula with the three clauses \(c_1=(x_1, x_2, x_3)\), \(c_2=(x_1, \lnot x_2, x_4)\), and \(c_3=(x_1, x_3, \lnot x_4)\).
      Shaded flags correspond to literals which do not appear in a clause.
      For the poles: thicker means positive part, thinner means negative part.
      Note that in both cases there is one pole which is flipped.
      }%
  \label{fig:logic-engine-concept}
\end{figure}
We add a \emph{flag} to the \(i\)th pole on the \(j\)th level as follows:
\begin{enumerate*}
  \item If \(x_i\) appears as \(x_i\) in \(c_j\) we add a flag on the negative part,
  \item if \(x_i\) appears as \(\lnot x_i\) in \(c_j\) we add a flag on the positive part, and
  \item if \(x_i\) does not appear in \(c_j\) we add a flag on both parts (hatched in \cref{fig:logic-engine-concept}).
\end{enumerate*}
Two vertical poles are added, one on the left and one on the right.
Note that in both realizations of \cref{fig:logic-engine-concept} there is one pole which is flipped.
Otherwise it would not be drawable without overlap.

The question is now: Can the logic engine be drawn without overlap?
For the drawing the variable poles can be flipped along their center and the flags can be drawn either left or right.
In a non-overlapping drawing the leftmost pole puts its flags to the right and the rightmost one puts its flags to the left.
Hence, every level \(j\) (top \emph{and} bottom) needs at least one pole \(i\) \emph{without} a flag on this level.
The corresponding literal of \(x_i\) appears in \(c_j\), fulfilling \(c_j\).
There cannot be a clause \(c_j\) with only positive literals---then the \(j\)th level on the bottom would have a flag on every variable pole; this is impossible without overlap.
The upper part of the drawing finds a positive literal for each clause and the lower part finds a negative literal.
Whether a variable pole was flipped for the drawing gives a direct correspondence to the assignment of its corresponding variable to \(1\) (not flipped) or \(0\) (flipped).
Hence, the logic engine can be drawn without any overlaps if and only if the corresponding NAE3SAT formula is satisfiable.
Constructing a logic engine with a WUDCR of trees gives a direct reduction from NAE3SAT and thus shows \NP-hardness.

\subsection{Rigid Hexagons as Basic Building Blocks}%
\label{sec:rigid-hexagons}

The goal is to model the logic engine structure by weak unit disk representations of trees.
The weak model allows us to tightly pack disks, something we use heavily.
The whole construction will live on a hexagonal grid with distance \(2\) (the diameter of a unit disk) between the grid points.
The \emph{grid distance} between two grid points is the number of edges on a shortest path---two touching disks have grid distance \(1\).

\begin{figure}[t]
    \includegraphics[scale=0.85,page=2,clip]{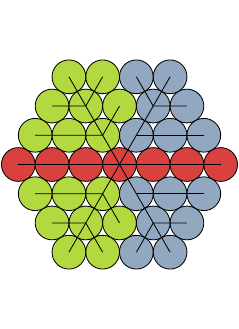}%
    \hfill%
    \includegraphics[scale=0.85,page=4,clip]{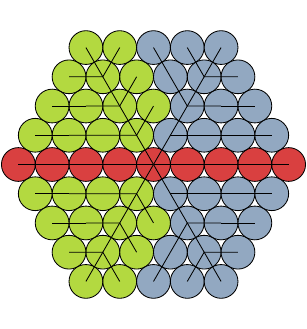}%
    \hfill%
    \includegraphics[scale=0.85,page=6]{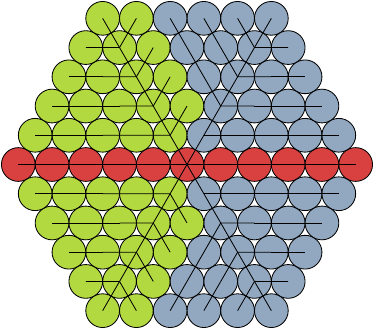}%
    \caption{Hexagons with \(r=3,4,5\). Hexagons with arbitrary radii \(r\geq 3\) can be constructed.}%
    \label{fig:hexagon-examples}
\end{figure}

We will construct a tree which can only be realized as a hexagon and which can be chained together to form longer and rigid structures.
For the chaining we need two special vertices which will always be on opposite corners of the resulting hexagon.
In \cref{fig:hexagon-examples} there are examples with various radii \(r\) (maximal grid distance to the center) which fulfill this criterion, as will be shown in
\begin{lemma}%
    \label[lemma]{lem:rigid-hexagons}
    All possible WUDCR of the trees in \cref{fig:hexagon-examples} are hexagons where the (red) paths end on opposite corners of the hexagon.
\end{lemma}
\begin{proof}
    We show that the trees are always hexagons and that all red nodes lie on a line.

    Observe that in \cref{fig:hexagon-examples} a tree node has distance \(k\) to the root if and only if its corresponding disks has grid distance \(k\) from the center disk.
    Hence, we have exactly as many nodes with distance \(k\) from the root as we have positions with grid distance \(k\) from some fixed location.
    The root node with only its direct children is realized as a disk with six neighboring disks.
    This is a tight packing and w.l.o.g.\ we can assume that they lie on a hexagonal grid.
    Furthermore, all but the last level of the tree have at least one node with \(3\) children (green nodes in \cref{fig:hexagon-examples}); this forces them onto the hexagonal grid: \(3\) of the \(6\) neighboring positions are taken by the parent and two siblings, the other three by the children.
    As a result all nodes on this level are forced onto positions on the hexagonal grid and the result is a tight packing of disks which forms a hexagon.

    Assume that not all red nodes are placed on a line.
    Look at \cref{fig:hexagon-path-has-to-be-straight} where four such situations are depicted.
    Due to the structure of the green and blue subtrees, placing the disks as in \cref{fig:hexagon-path-has-to-be-straight:no-bend} leaves at least one grid position \(c\) empty.
    Placing the disks as in \cref{fig:hexagon-path-has-to-be-straight:bend} leaves the corner position \(c\) empty.
    A shortest path from the center disk to \(c\) (shown dotted) has distance \(k+1\) if \(k\) is \(c\)'s grid distance to the center.
    There is no node with depth \(k+1\) in the tree---\(c\) is left empty.
    However, as all grid positions with grid distance up to the tree's height are covered, there is at least one node whose disk cannot be placed without overlap.

    We conclude that all red nodes are forced on a line which places the two leaf nodes on opposite corners of the hexagon.
\end{proof}

We say that the trees for the hexagons have only one \emph{distinguishable} WUDCR:
It means that the placement of the important nodes (where something else will be connected to) does not change, but the placement of the other nodes may.

\begin{figure}[t]
    \begin{subfigure}[t]{0.47\linewidth}
        \centering
        \includegraphics[width=\textwidth,page=12]{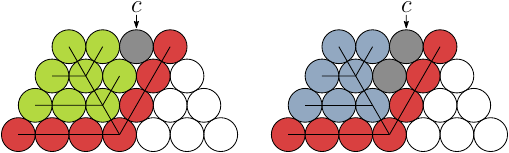}
        \caption{Placing path segments on individual lines leaves at least one uncovered grid position \(c\).}%
        \label{fig:hexagon-path-has-to-be-straight:no-bend}
    \end{subfigure}
    \hfill%
    \begin{subfigure}[t]{0.47\linewidth}
        \centering
        \includegraphics[width=\textwidth,page=11]{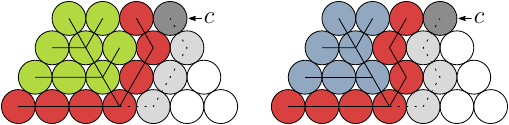}
        \caption{Bending a path segment gives \(c\) an unreachable shortest path distance of \(4\).}%
        \label{fig:hexagon-path-has-to-be-straight:bend}
    \end{subfigure}
    \caption{Not placing both path segments on a common line leaves unreachable grid positions.}%
    \label{fig:hexagon-path-has-to-be-straight}
\end{figure}

\begin{figure}[t]
    \hfill%
    \includegraphics[scale=1,page=7]{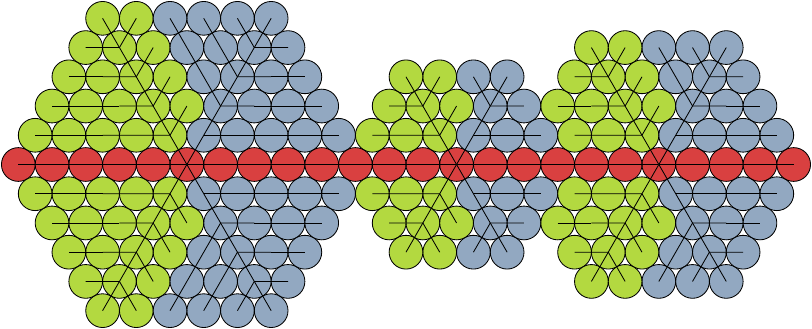}
    \hfill\null%
    \caption{%
        Connecting multiple hexagons, which can have different sizes.
        Observe, that the hexagons can even overlap more than one disk (center and right).
    }%
    \label{fig:connecting-hexagons}
\end{figure}

Plugging two hexagons together at a common endpoint forces them to lie on the same line.
With this we are able to construct longer straight paths. 
Additionally, the connected hexagons can differ in size or they can overlap by more than just one disk.
See an example of three connected hexagons in \cref{fig:connecting-hexagons}.

\subsection{A Branching Gadget}%
\label{sec:branching}

Apart from going in a straight line, we need to be able to branch off two \emph{branches} from a straight part (the \emph{trunk}) to simulate the variable poles and flags.
Additionally, it must allow for both branches to be interchanged to flip the variable poles or flags between two sides.
As we could not branch orthogonally in a fully rigid way we will instead introduce a branching gadget which branches off at a \SI{60}{\degree} angle, see \cref{fig:branching-gadget}.

\begin{figure}[t]
    \hfill%
    \includegraphics[scale=0.5,page=8]{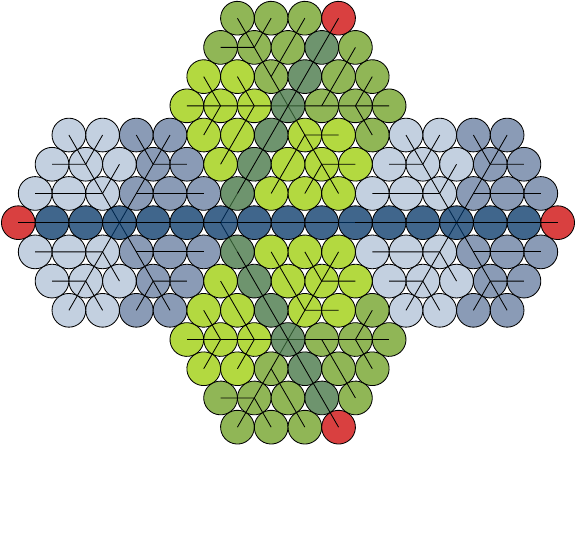}%
    \hfill%
    \includegraphics[scale=0.5,page=9]{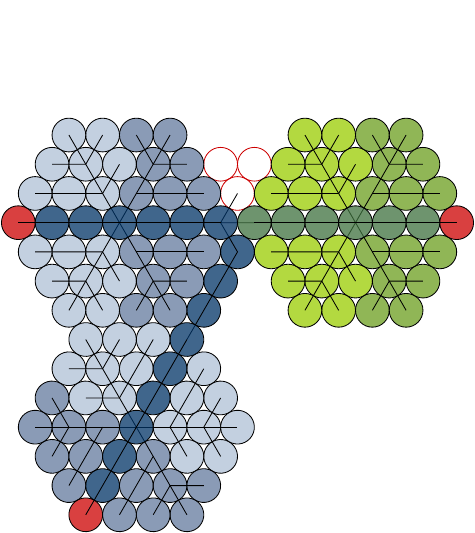}%
    \hfill%
    \includegraphics[scale=0.5,page=10]{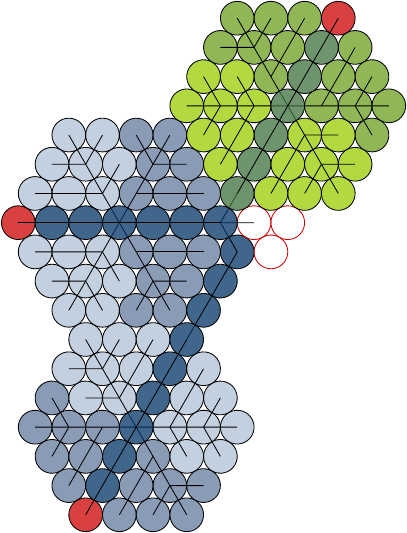}%
    \hfill\null%
    \caption{
        The branching gadget with \SI{60}{\degree} angle and interchangeable sides (left).
        Placing the horizontal part non-horizontally does not leave sufficient room for both branches (center, right).
    }%
    \label{fig:branching-gadget}
\end{figure}

\begin{obs}%
    \label[obs]{obs:branching-gadget-two-realizations}
    The branching gadget in \cref{fig:branching-gadget} (left) has exactly two WUDCR which differ in the placement of the red vertices.
\end{obs}
\begin{proof}
    From \cref{lem:rigid-hexagons} we know that the four hexagons are rigid.
    Two hexagons and one path (ending in a hexagon) are connected to the leftmost hexagon.
    As shown in \cref{fig:branching-gadget} (left) it is possible to place the path horizontally.
    Placing the path differently, e.g.\ as in \cref{fig:branching-gadget} (center and right), leaves no space to fit both hexagons.
    Due to symmetry, both branching hexagons can be interchanged in the left case, giving two distinguishable WUDCR.
\end{proof}

\subsection{Simulating the Logic Engine}%
\label{sec:logic-engine-simulation}

\begin{figure}[t]
  \begin{subfigure}[t]{0.48\linewidth}
    \centering
    \includegraphics[width=1\textwidth,page=1]{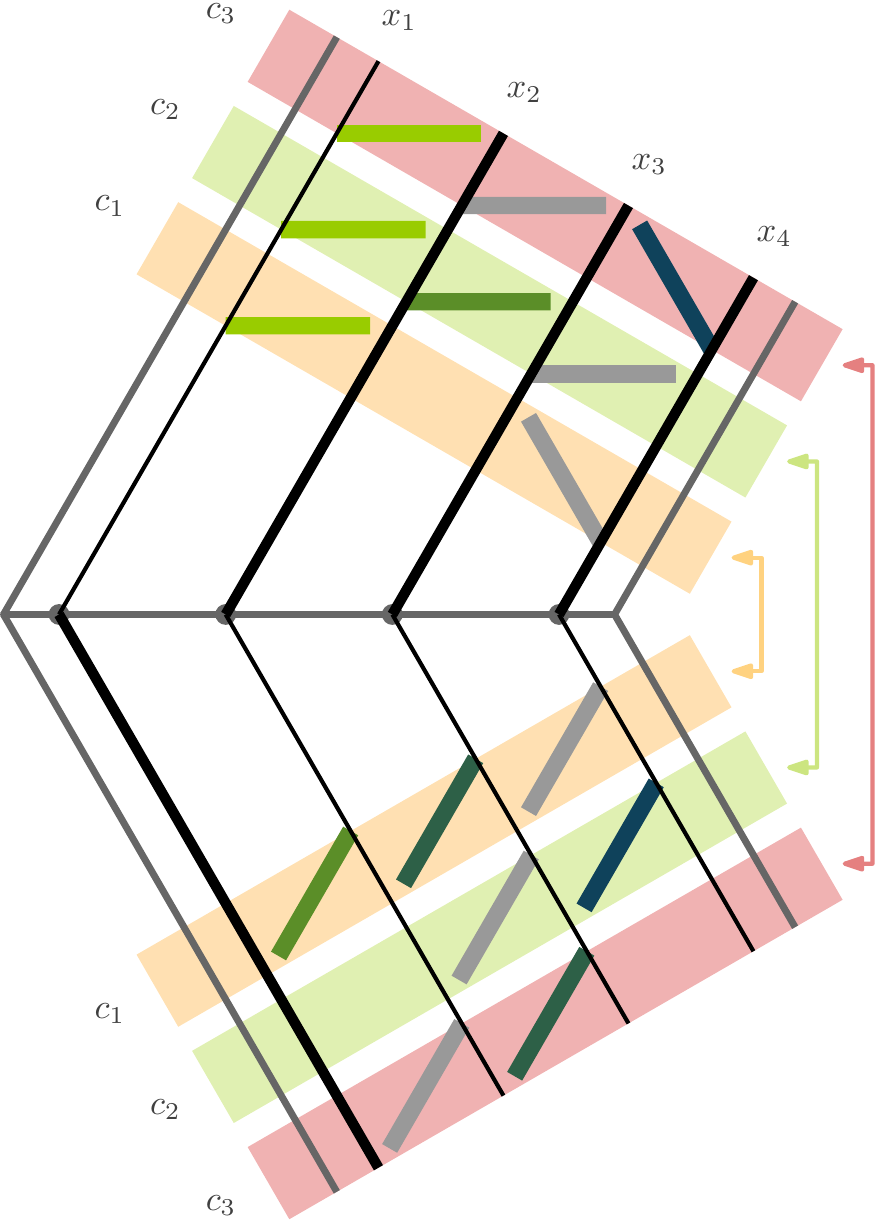}
    \caption{\(x_1=0\), \(x_2=1\), \(x_3=1\), and \(x_4=1\).}
  \end{subfigure}
  \hfill%
  \begin{subfigure}[t]{0.48\linewidth}
    \centering
    \includegraphics[width=1\textwidth,page=2]{logic-engine-tilted.pdf}
    \caption{\(x_1=1\), \(x_2=1\), \(x_3=0\), and \(x_4=1\).}
  \end{subfigure}
  \caption{%
      Two logic engine realizations for the clauses \(c_1=(x_1, x_2, x_3)\), \(c_2=(x_1, \lnot x_2, x_4)\), and \(c_3=(x_1, x_3, \lnot x_4)\).
      It is a modification of \cref{fig:logic-engine-concept} where the branching angles are \SI{60}{\degree} instead of \SI{90}{\degree}.
  }%
  \label{fig:logic-engine-concept-tilted}
\end{figure}

The logic engine needs orthogonal branching and we only have branches at at \SI{60}{\degree} angle.
Hence, it is necessary to modify the logic engine in a way to accommodate for such a difference.
\cref{fig:logic-engine-concept-tilted} shows a modification of \cref{fig:logic-engine-concept} which only includes \SI{60}{\degree} angles.
Flipping of the poles and flags happens by mirroring them along the line segment they are attached to.
As can be seen, the clauses are still orthogonal to the variable poles s.t.\ two flags in the same free space are forced to overlap.

\begin{figure}[p]
    \hfill%
    \includegraphics[width=\textwidth,page=11]{whole-picture.pdf}
    \hfill\null%
    \caption{%
        An example of a logic engine for the formula with clauses \(c_1=(x_1, x_2, x_3)\), \(c_2=(x_1, \lnot x_2, x_4)\), and \(c_3=(x_1, x_3, \lnot x_4)\).
        The variables are set to \(x_1=0\), \(x_2=1\), \(x_3=1\), and \(x_4=1\).
    }%
    \label{fig:whole-picture}
\end{figure}

\begin{figure}[tbp]
    \hfill%
    \includegraphics[width=0.5\textwidth,page=12]{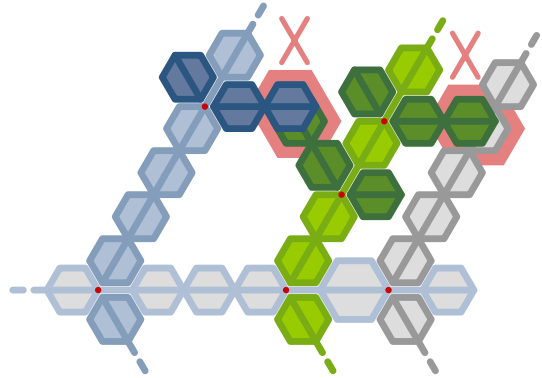}
    \hfill\null%
    \caption{%
        Possible cases of overlaps are highlighted:
        No two flags can be in the same place (left) and all flags on the outer variable poles must face inside (right).
    }%
    \label{fig:reduction-details}
\end{figure}

\begin{theorem}%
    \label[theorem]{thm:np-hardness-of-recognizing-trees}
    It is \NP-hard to decide whether a tree has weak unit disk contact representation.
\end{theorem}
\begin{proof}
    We model the logic engine from \cref{fig:logic-engine-concept-tilted} with the hexagon and branching gadgets to reduce NAE3SAT to our problem.
    See \cref{fig:whole-picture} for an example of how a resulting drawing might look like.
    Apart from the two hexagons with radius \(4\) near the left and right end of the horizontal spine we only use hexagons with radius \(3\) and the branching gadget from before (also with radius \(3\) hexagons).

    As shown in \cref{fig:reduction-details} the distance between two variable poles and placement of the branching gadgets enforces that no two flags can be placed into the same free space.
    Furthermore, the left and right frame prevent flags from being drawn to the outside.
    The tree constructed from a boolean formula has a WUDCR if and only if no overlap occurs.
    This happens if and only if for every level \(j\) / clause \(c_j\) there is at least one flag less than the total number of variables; or to put it differently: if and only if not all three variables appearing in \(c_j\) place a flag on the bottom and not all on the top.
    This then gives a satisfying assignment of variables where not all literals evaluate to \(1\).

    The size of the construction is polynomial in the number of clauses \(m\) and variables \(n\).
    There is a constant distance between two variable poles, hence, the size of the horizontal spine is \(O(n)\).
    The further left a variable pole is the longer it has to be.
    The part without branching grows linearly in \(n\) (\(2\) more hexagons per step to the left) and the branching part grows linearly in \(m\), since each branching gadget with flag has constant size: the total size of one variable pole is \(O(n+m)\).
    For \(n\) variable poles (and the two frames) this gives a total size of \(O(n^2+mn)\) for the whole construction and it can be easily constructed in polynomial time.
\end{proof}

\section{Conclusion}

We showed that in linear time we can decide whether a caterpillar graph can be realized as a weak unit disk contact representation.
On the other hand, the same problem is \NP-hard for trees.
The main open question remains whether lobster graphs (every node has distance at most \(2\) to a central path) can be recognized in polynomial time or whether it is \NP-hard to recognize them.
This can be generalized to look at trees where each node has a distance at most \(d\) from a central path.

\subparagraph*{Acknowledgments.}
The author wants to give special thanks to Sujoy Bhore, Man-Kwun Chiu, Soeren Nickel, and Martin Nöllenburg for the discussions and ideas during a research visit in Vienna which laid the groundwork for this paper.

\bibliography{weak-udcr}

\end{document}